\documentclass[11pt]{article}
\usepackage{amsfonts}
\usepackage{amssymb}
\usepackage{amsmath}
\usepackage{a4wide}
\usepackage{graphicx}

\newtheorem{theorem}{Theorem}

\newtheorem{corollary}{Corollary}

\newtheorem{lemma}{Lemma}

\newtheorem{proposition}{Proposition}
\newtheorem{remark}{Remark}

\newenvironment{proof}[1][Proof]{\textbf{#1.} }{\ \rule{0.5em}{0.5em}}
\makeatletter
\def\@removefromreset#1#2{\let\@tempb\@elt
     \def\@tempa#1{@&#1}\expandafter\let\csname @*#1*\endcsname\@tempa
     \def\@elt##1{\expandafter\ifx\csname @*##1*\endcsname\@tempa\else
    \noexpand\@elt{##1}\fi}     \expandafter\edef\csname cl@#2\endcsname{\csname cl@#2\endcsname}     \let\@elt\@tempb
     \expandafter\let\csname @*#1*\endcsname\@undefined}

\@removefromreset{equation}{section}

\@removefromreset{theorem}{section}
\makeatother
\input{tcilatex}
\begin{document}

\title{All joint von Neumann measurements on a quantum state admit a
quasi-classical probability model}
\author{Elena R. Loubenets \\
Moscow State Institute of Electronics and Mathematics, Moscow 109028,
Russia\medskip }
\maketitle

\begin{abstract}
We prove that the Hilbert space description of \emph{all} joint von Neumann
measurements on a quantum state can be reproduced in terms of a \emph{single}
measure space $(%
\Omega
,\mathcal{F}_{\Omega },%
\mu
)$ with a normalized real-valued measure $%
\mu
$, \ that is, in terms of a new general probability model, \emph{the
quasi-classical probability model, }developed in [Loubenets: \emph{J. Math.
Phys}. 53 (2012), 022201; \emph{J. Phys. A: Math. Theor}. 45 (2012),
185306]. In a quasi-classical probability model for all von Neumann
measurements, a random variable models the corresponding quantum observable
in all joint measurements and depends only on this quantum observable. This
mathematical result sheds a new light on some important issues of quantum
randomness discussed in the literature since the seminal article (1935) of
Einstein, Podolsky and Rosen.
\end{abstract}

\tableofcontents

\section{Introduction}

The relation between the quantum probability model and the classical
probability model has been a point of intensive discussions ever since the
seminal publications of von Neumann \cite{1}, Kolmogorov \cite{7} and \
Einstein, Podolsky and Rosen (EPR) \cite{2}.

Though, in the quantum physics literature, one can still find the misleading
claims on a peculiarity of \textquotedblright quantum probabilities" and
\textquotedblright quantum events", the probabilistic description of every
quantum measurement, generalized \cite{9} or projective, satisfies the
Kolmogorov axioms \cite{7} for the theory of probability. For example, the
von Neumann measurement of a quantum observable $X$ in a state $\rho $ on a
complex separable Hilbert space $\mathcal{H}$ is described by the
probability space $(\mathbb{R},\mathcal{B}_{\mathbb{R}},\mathrm{tr}[\rho 
\mathrm{P}_{X}(\cdot )]),$ where $\mathcal{B}_{\mathbb{R}}$ is the Borel $%
\sigma $-algebra on $\mathbb{R}$ and $\mathrm{P}_{X}$ is the spectral
projection-valued measure on $\mathcal{B}_{\mathbb{R}}$ uniquely
corresponding to an observable $X$ on $\mathcal{H}$ due the the spectral
theorem \cite{1, 14}.

However, the Hilbert space description of \emph{all} joint von Neumann
measurements on an arbitrary quantum state cannot be reproduced in terms of
a single probability space. The same concerns the probabilistic description
of an arbitrary quantum multipartite correlation scenario with finite
numbers of settings at each site and, more generally, an arbitrary
nonsignaling multipartite correlation scenario, see introductions in \cite%
{10, 11, 12} and references therein.

Note that, in the quantum theory literature, the interpretation of quantum
measurements via the classical probability model is generally referred to as
a hidden variable (HV) model, and a local hidden variable (LHV) model
constitutes a version of an HV model, where a random variable modelling a
marginal measurement depends only on a setting of this measurement.

Analyzing the probabilistic description of a \emph{general} multipartite
correlation scenario with a finite number of settings at each site, we have
introduced \cite{11} the notion of \emph{a} \emph{local quasi hidden
variable (LqHV) model, }where locality and the measure-theoretic structure
inherent to a local hidden variable (LHV) model are preserved but positivity
of a simulation measure is dropped. We have proved \cite{11} that every
quantum $N$-partite correlation scenario admits an LqHV model.

Developing the LqHV approach further, we showed \cite{12} that a general
correlation scenario admits an LqHV model (i) if and only if it is
nonsignaling \cite{10} and (ii) if and only if it admits a deterministic
LqHV model. In the latter particular type of an LqHV model, all joint
probability distributions of a correlation scenario are reproduced in terms
of a single measure space $(\Omega ,\mathcal{F}_{\Omega },\mu )$ with a
normalized bounded measure $\mu $ and a set of random variables each
depending only on a setting of the corresponding modelled marginal
measurement.

As we have argued in \cite{12}, these new results point to the existence of
a new general probability model, \emph{the quasi-classical probability
model, }that has the measure-theoretic structure $(\Omega ,\mathcal{F}%
_{\Omega },\nu )$ resembling the structure of the classical probability
model but reduces to the latter iff a normalized real-valued measure $\nu $
is positive.

In the present paper, we prove that the Hilbert space description of \emph{%
all joint von Neumann measurements} on a quantum system can be reproduced in
terms of a \emph{single} space $(\Omega ,\mathcal{F}_{\Omega })$ via a set
of normalized real-valued measures, each uniquely corresponding to a
modelled quantum state, and a set of random variables, each modelling the
corresponding quantum observable in all joint von Neumann measurements and
depending only on this quantum observable.

This result, in particular, means that the probabilistic description of all
joint von Neumann measurements on a quantum state admits the quasi-classical
probability model introduced in \cite{12} and, in this model, a random
variable modelling a quantum observable in all joint von Neumann
measurements is determined only by this quantum observable.

The paper is organized as follows.

In section 2, we recall the von Neumann formalism for the description of
ideal quantum measurements and the notion of the spectral projection-valued
measure of a quantum observable.

In section 3, we introduce for a finite number of quantum observables the
symmetrized product of their spectral measures and discuss properties of
this product operator-valued measure.

In section 4, we generalize some items of the Kolmogorov extension theorem 
\cite{7} to the case of consistent operator-valued measures and specify this
generalization for the consistent product measures introduced in section 3.

In section 5, we prove that the Hilbert space description of all joint von
Neumann measurements on a quantum system can be reproduced in terms of
random variables and normalized real-valued measures defined on a single
measurable space.

In section 6, we summarize the main mathematical results of the present
paper and discuss their conceptual implications.

\section{Von Neumann measurements}

In the frame of the von Neumann approach \cite{1} , states and observables
of a quantum system are described, correspondingly, by density operators $%
\rho $ and self-adjoint linear operators $X,$ bounded or unbounded, on a
complex separable Hilbert space $\mathcal{H}$.

Denote by $\mathfrak{X}_{\mathcal{H}}$ the set of all self-adjoint linear
operators, bounded and unbounded, on $\mathcal{H}$. Let $\mathcal{L}_{%
\mathcal{H}}$ be the vector space of all bounded linear operators on $%
\mathcal{H}$ and $\mathcal{L}_{\mathcal{H}}^{(s)}$\ -- the vector space of
all self-adjoint bounded linear operators on $\mathcal{H}.$ Equipped with
the operator norm, these vector spaces are Banach.

The probability that, under an ideal (errorless) measurement of a quantum
observable $X\in \mathfrak{X}_{\mathcal{H}}$ in a state $\rho ,$ an observed
value belongs to a Borel subset $B$ of $\mathbb{R}$ is given \cite{1, 9, 14}
by the expression%
\begin{equation}
\mathrm{tr}[\rho \mathrm{P}_{X}(B)],\text{ \ \ }B\in \mathcal{B}_{\mathbb{R}%
},  \label{1}
\end{equation}%
where $\mathcal{B}_{\mathbb{R}}$ is the $\sigma $-algebra \cite{13} of Borel
subsets of $\mathbb{R}$ and $\mathrm{P}_{X}$ is the spectral
projection-valued measure of a quantum observable $X$ -- that is, the
measure $\mathrm{P}_{X}$ on $\mathcal{B}_{\mathbb{R}}$ uniquely
corresponding to $X\in \mathfrak{X}_{\mathcal{H}}\ $due to the spectral
theorem \cite{1, 14, 15} 
\begin{equation}
X=\dint\limits_{\mathbb{R}}\lambda \mathrm{P}_{X}(\mathrm{d}\lambda )
\label{2}
\end{equation}%
and with values $\mathrm{P}_{X}(B),$ $\forall B\in \mathcal{B}_{\mathbb{R}},$
$\mathrm{P}_{X}(\mathbb{R})=\mathbb{I}_{\mathcal{H}},$ that are projections
on $\mathcal{H}$ satisfying the relations 
\begin{eqnarray}
\mathrm{P}_{X}(B_{1})\mathrm{P}_{X}(B_{2}) &=&\mathrm{P}_{X}(B_{2})\mathrm{P}%
_{X}(B_{1})=\mathrm{P}_{X}(B_{1}\cap B_{2}),\text{ \ \ }B_{1},B_{2}\in 
\mathcal{B}_{\mathbb{R}},  \label{3} \\
\mathrm{P}_{X}(B) &=&0,\text{\ \ iff \ }B\in \mathcal{B}_{\mathbb{R}}\cap (%
\mathbb{R}/\mathrm{sp}X),  \notag
\end{eqnarray}%
where the spectrum $\mathrm{sp}X$ of an observable $X\in \mathfrak{X}_{%
\mathcal{H}}$ constitutes a closed Borel subset of $\mathbb{R}.$

An ideal measurement (\ref{1}) of a quantum observable $X$ in a state $\rho $
is generally referred to as the von Neumann measurement.

The measure $\mathrm{P}_{X}$ is $\sigma $-additive in the strong operator
topology \cite{9, 14, 15} in $\mathcal{L}_{\mathcal{H}}^{(s)}$, that is:%
\begin{equation}
\lim_{n\rightarrow \infty }\left\Vert \mathrm{P}_{X}(\cup _{i=1}^{\infty
}B_{i})\psi -\sum_{i=1}^{n}\mathrm{P}_{X}(B_{i})\psi \right\Vert _{\mathcal{H%
}}=0  \label{4}
\end{equation}%
for all $\psi \in \mathcal{H}$ and all countable collections $\{B_{i}\}$ of
mutually disjoint sets in $\mathcal{B}_{\mathbb{R}}.$

\begin{remark}
In this article, we follow the terminology of Ref. \cite{16}. Namely, let $%
\mathfrak{B}$ be a Banach space and $\mathcal{F}_{\Lambda }$ be an algebra
of subsets of a set $\Lambda .$ We refer to an additive set function $%
\mathfrak{m}:\mathcal{F}_{\Lambda }\rightarrow \mathfrak{B}$ as a $\mathfrak{%
B}$-valued (finitely additive) measure on $\mathcal{F}_{\Lambda }$. If a
measure $\mathfrak{m}$ on $\mathcal{F}_{\Lambda }$ is $\sigma $-additive in
some topology on $\mathfrak{B},$ then we specify this in addition.
\end{remark}

From (\ref{3}) it follows that, for each $X\in \mathfrak{X}_{\mathcal{H}},$
the measure $\mathrm{P}_{X}(B)\neq 0$ if and only if a set $B\neq
\varnothing $ belongs to the trace $\sigma $-algebra 
\begin{equation}
\mathcal{B}_{\mathrm{sp}X}:=\mathcal{B}_{\mathbb{R}}\cap \mathrm{sp}X.
\end{equation}%
Therefore, we further consider the spectral projection-valued measure $%
\mathrm{P}_{X}$ only on $\mathcal{B}_{\mathrm{sp}X}.$

The joint von Neumann measurement of several quantum observables $%
X_{1},...,X_{n}\in \mathfrak{X}_{\mathcal{H}}$ is possible \cite{1, 9, 14}
iff all values of their spectral measures mutually commute, that is, 
\begin{equation}
\lbrack \mathrm{P}_{X_{i_{1}}}(B_{i_{1}}),\mathrm{P}%
_{X_{i_{2}}}(B_{i_{2}})]=0,\text{ \ \ }B_{i}\in \mathcal{B}_{\mathrm{sp}%
X_{i}},\text{ \ \ }i=1,...,n.  \label{com}
\end{equation}

For bounded quantum observables $X_{1},...,X_{n}\in \mathfrak{X}_{\mathcal{H}%
}$, condition (\ref{com}) is equivalent to mutual commutativity $%
[X_{i_{1}},X_{i_{2}}]=0,$ $i=1,...,n,$ of these observables. Therefore, for
short, we further refer to arbitrary quantum observables $X_{1},...,X_{n}\in 
\mathfrak{X}_{\mathcal{H}}$, for which the spectral measures satisfy
condition (\ref{com}), as \emph{mutually commuting }in the sense of
condition (\ref{com}).

The joint von Neumann measurement of \emph{mutually commuting} quantum
observables $X_{1},...,X_{n}\in \mathfrak{X}_{\mathcal{H}}$ is described 
\cite{1, 9, 14} by the normalized projection-valued measure 
\begin{equation}
\dint\limits_{(\lambda _{1},...,\lambda _{n})\in B}\mathrm{P}_{X_{1}}(%
\mathrm{d}\lambda _{1})\cdot ...\cdot \mathrm{P}_{X_{n}}(\mathrm{d}\lambda
_{n}),\text{ \ \ \ \ }B\in \mathcal{B}_{\mathrm{sp}X_{1}\times \cdots \times 
\mathrm{sp}X_{n}},  \label{5}
\end{equation}%
on the trace Borel $\sigma $-algebra 
\begin{equation}
\mathcal{B}_{\mathrm{sp}X_{1}\times \cdots \times \mathrm{sp}X_{n}}:=%
\mathcal{B}_{\mathbb{R}^{n}}\cap (\mathrm{sp}X_{1}\times \cdots \times 
\mathrm{sp}X_{n}).  \label{a}
\end{equation}%
This measure is \cite{15} $\sigma $-additive in the strong operator topology
on $\mathcal{L}_{\mathcal{H}}^{(s)}$.

The expression 
\begin{equation}
\mathrm{tr}[\rho \{\mathrm{P}_{X_{1}}(B_{1})\cdot ...\cdot \mathrm{P}%
_{X_{n}}(B_{n})\}]  \label{6}
\end{equation}%
gives the probability that, under the joint von Neumann measurement of
mutually commuting quantum observables $X_{1},...,X_{n}\in \mathfrak{X}_{%
\mathcal{H}}$ in a state $\rho $, these observables take values in sets $%
B_{1}\in \mathcal{B}_{\mathrm{sp}X_{1}},...,B_{n}\in \mathcal{B}_{\mathrm{sp}%
X_{n}},$ respectively.

\section{Symmetrized products of spectral measures}

For an $n$-tuple $(X_{1},...,X_{n})$ of arbitrary mutually non-equal
observables $X_{1},...,X_{n}$ $\in \mathfrak{X}_{\mathcal{H}},$ consider on
the set $\mathrm{sp}X_{1}\times \cdots \times \mathrm{sp}X_{n}\subseteq 
\mathbb{R}^{n}$ the algebra $\mathcal{F}_{\mathrm{sp}X_{1}\times \cdots
\times \mathrm{sp}X_{n}}$, \emph{the product algebra}, generated by all
rectangles $B_{1}\times \cdots \times B_{n}\subseteq \mathrm{sp}X_{1}\times
\cdots \times \mathrm{sp}X_{n}$ with measurable sides $B_{i}\in \mathcal{B}_{%
\mathrm{sp}X_{i}}.$

Let 
\begin{equation}
\mathcal{P}_{(X_{1},...,X_{n})}:\mathcal{F}_{\mathrm{sp}X_{1}\times \cdots
\times \mathrm{sp}X_{n}}\rightarrow \mathcal{L}_{\mathcal{H}}^{(s)}
\label{me}
\end{equation}%
be the normalized finitely additive $\mathcal{L}_{\mathcal{H}}^{(s)}$-valued
measure defined \emph{uniquely} on $\mathcal{F}_{\mathrm{sp}X_{1}\times
\cdots \times \mathrm{sp}X_{n}}$ via the relation 
\begin{equation}
\mathcal{P}_{(X_{1},...,X_{n})}(B_{1}\times \cdots \times B_{n})=\frac{1}{n!}%
\left\{ \mathrm{P}_{X_{1}}(B_{1})\cdot \ldots \cdot \mathrm{P}%
_{X_{n}}(B_{n})\right\} _{\mathrm{sym}}  \label{9_}
\end{equation}%
on all rectangles $B_{1}\times \cdots \times B_{n}$ with sides $B_{i}\in 
\mathcal{B}_{\mathrm{sp}X_{i}}.$ Here, the notation $\{Z_{1}\cdot \ldots
\cdot Z_{n}\}_{\mathrm{sym}}$ means the sum constituting the symmetrization
of the operator product $Z_{1}\cdot \ldots \cdot Z_{n},$ $Z_{i}\in \mathcal{L%
}_{\mathcal{H}}^{(s)},$ with respect to all permutations of its
factors.\medskip

If each observable $X_{i}$ in a collection $\{X_{1},...,X_{n}\}\subset 
\mathfrak{X}_{\mathcal{H}}$ is bounded (i.e. $X_{i}\in \mathcal{L}_{\mathcal{%
H}}^{(s)}$) and has only a discrete spectrum $\mathrm{sp}X_{i}=\{\lambda
_{X_{i}}^{(k)}\in \mathbb{R},$ $k=1,...,K_{X_{i}}<\infty \},$ where $\lambda
_{X_{i}}^{(k)}$ are eigenvalues of $X_{i},$ then the product algebra $%
\mathcal{F}_{\mathrm{sp}X_{1}\times \cdots \times \mathrm{sp}X_{n}}$ is
finite and coincides with the Borel algebra $\mathcal{B}_{\mathrm{sp}%
X_{1}\times \cdots \times \mathrm{sp}X_{n}}$, and the finitely additive
measure $\mathcal{P}_{(X_{1},...,X_{n})}$ has the form\footnote{%
Here, the generally accepted notation $\sum_{\lambda \in B}Z(\lambda )$
means the sum $\sum_{\lambda }\chi _{B}(\lambda )Z(\lambda ),$ where $\chi
_{B}(\cdot )$ is the indicator function of a set $B,$ that is, $\chi
_{B}(\lambda )=1$ for $\lambda \in B$ and $\chi _{B}(\lambda )=0$ for $%
\lambda \notin B.$} 
\begin{equation}
\mathcal{P}_{(X_{1},...,X_{n})}(F):=\frac{1}{n!}\dsum\limits_{(\lambda
_{X_{1}},...,\lambda _{X_{n}})\in F}\left\{ \text{ }\mathrm{P}%
_{X_{1}}(\{\lambda _{X_{1}}\})\cdot \ldots \cdot \mathrm{P}%
_{X_{n}}(\{\lambda _{X_{n}}\})\right\} _{\mathrm{sym}}  \label{9}
\end{equation}%
for all $F\in \mathcal{F}_{\mathrm{sp}X_{1}\times \cdots \times \mathrm{sp}%
X_{n}}$.

For quantum observables $X_{1},...,X_{n}\in \mathfrak{X}_{\mathcal{H}}$,
which mutually commute in the sense of relation (\ref{com}), the measure $%
\mathcal{P}_{(X_{1},...,X_{n})}$ is projection-valued and $\left\Vert 
\mathcal{P}_{(X_{1},...,X_{n})}(F)\right\Vert $ $=1$ for each $\varnothing
\neq F\in \mathcal{F}_{\mathrm{sp}X_{1}\times \cdots \times \mathrm{sp}%
X_{n}} $

Consider the family%
\begin{equation}
\{\mathcal{P}_{(X_{1},...,X_{n})}:\mathcal{F}_{\mathrm{sp}X_{1}\times \cdots
\times \mathrm{sp}X_{n}}\rightarrow \mathcal{L}_{\mathcal{H}}^{(s)}\mid
\{X_{1},...,X_{n}\}\subset \mathfrak{X}_{\mathcal{H}},\text{ }n\in \mathbb{N}%
\}  \label{12}
\end{equation}%
of all normalized finitely additive $\mathcal{L}_{\mathcal{H}}^{(s)}$-valued
measures (\ref{me}). These measures satisfy the following relations proved
in appendix A.

\begin{lemma}
For an arbitrary finite collection $\{X_{1},...,X_{n}\}$ $\subset \mathfrak{X%
}_{\mathcal{H}}$ of quantum observables on $\mathcal{H},$ 
\begin{eqnarray}
\mathcal{P}_{(X_{1},...,X_{n})}(B_{1}\times \cdots \times B_{n}) &=&\mathcal{%
P}_{(X_{i_{1}},...,X_{i_{_{n}}})}(B_{i_{1}}\times \cdots \times B_{i_{n}}),
\label{13} \\
B_{i} &\in &\mathcal{B}_{\mathrm{sp}X_{i}},\text{ \ \ }i=1,...,n,  \notag
\end{eqnarray}%
for all permutations $\binom{1,,...,n}{i_{1},...,i_{n}}$ and 
\begin{eqnarray}
&&\mathcal{P}_{(X_{1},...,X_{_{n}})}\left( \{(x_{1},...,x_{n})\in \mathrm{sp}%
X_{1}\times \cdots \times \mathrm{sp}X_{n}\mid (x_{i_{1}},...,x_{i_{k}})\in
F\}\right)  \label{14} \\
&=&\mathcal{P}_{(X_{i_{1}},...,X_{i_{k}})}\left( F\right) ,\text{ \ \ \ }%
F\in \mathcal{F}_{\mathrm{sp}X_{i_{1}}\times \cdots \times \mathrm{sp}%
X_{i_{k}}},  \notag
\end{eqnarray}%
for each subset $\{X_{i_{1}},...,X_{i_{k}}\}\subseteq
\{X_{1},...,X_{n}\}.\medskip $
\end{lemma}

Relations (\ref{13}), (\ref{14}) on operator-valued measures $\mathcal{P}%
_{(X_{1},...,X_{_{n}})}$ are quite similar by their form to the Kolmogorov
consistency conditions \cite{7, 17} for a family 
\begin{equation}
\{\mu _{(t_{1},...,t_{n})}:\mathcal{B}_{\mathbb{R}^{n}}\rightarrow \lbrack
0,1]\text{ }\mid \{t_{1},...,t_{n}\}\subset T,\ \ n\in \mathbb{N}\}
\label{15}
\end{equation}%
of probability measures $\mu _{(t_{1},...,t_{n})}$, each specified by a
tuple $(t_{1},...,t_{n})$ of mutually non-equal elements in an index set $T.$

In view of this similarity, for our further consideration in section 5, we
proceed to generalize to the case of consistent operator-valued measures
some items of the Kolmogorov theorem \cite{7} on the extension to $(\mathbb{R%
}^{T},\mathcal{F}_{\mathbb{R}^{T}})$ of consistent probability measures (\ref%
{15}).

\begin{remark}
Notations $\mathbb{R}^{T}$ and $\mathcal{F}_{\mathbb{R}^{T}}$ mean \cite{7,
17}, correspondingly, the set of all real-valued functions $x:T\rightarrow 
\mathbb{R}$ and the algebra generated on $\mathbb{R}^{T}$ by all cylindrical
subsets of the form $\left\{ x\in \mathbb{R}^{T}\mid
(x(t_{_{1}}),...,x(t_{n}))\in B\right\} ,$ where $B\in \mathcal{B}_{\mathbb{R%
}^{n}},$ $\{t_{1},...,t_{n}\}$ $\subset T,$ $n\in \mathbb{N}.$
\end{remark}

\section{The extension theorem}

For an uncountable index set $T,$ consider a family $\{(\Lambda _{t},%
\mathcal{F}_{\Lambda _{t}}),$ $t\in T\},$ where each $\Lambda _{t}$ is a
non-empty set and $\mathcal{F}_{\Lambda _{t}}$ is an algebra of subsets of $%
\Lambda _{t}$. Let $\mathcal{F}_{\Lambda _{t_{1}}\times \cdots \times
\Lambda _{t_{n}}}$ be the algebra on $\Lambda _{t_{1}}\times \cdots \times
\Lambda _{t_{n}},$ \emph{the product algebra}, generated by all rectangles $%
F_{1}\times \cdots \times F_{n}\subseteq \Lambda _{t_{1}}\times \cdots
\times \Lambda _{t_{n}}$ with sides $F_{k}\in \mathcal{F}_{\Lambda
_{t_{k}}}. $

Denote by $\Lambda :=\dprod\limits_{t\in T}\Lambda _{t}$ the Cartesian
product \cite{13} of all sets $\Lambda _{t},$ $t\in T$. That is, $\Lambda $
is the collection of all functions $\lambda :T\rightarrow $ $\cup _{t\in
T}\Lambda _{t}$ with values $\lambda _{t}:=$ $\lambda (t)\in \Lambda _{t}.$

The set of all cylindrical subsets of $\Lambda $ of the form 
\begin{eqnarray}
\mathcal{J}_{(t_{1},...,t_{n})}(F) &:&=\left\{ \lambda \in \Lambda \mid
(\lambda _{t_{1}},...,\lambda _{t_{n}})\in F\right\} ,  \label{16} \\
F &\in &\mathcal{F}_{\Lambda _{t_{1}}\times \cdots \times \Lambda _{t_{n}}},%
\text{\ \ \ }\{t_{1},...,t_{n}\}\subset T,\text{\ \ \ }n\in \mathbb{N}, 
\notag
\end{eqnarray}%
constitutes \cite{13} an algebra on $\Lambda $ that we further denote by $%
\mathcal{A}_{\Lambda }.$

Since $\mathcal{J}_{(t_{1},...,t_{n})}(F)\equiv \pi
_{(t_{1},...,t_{n})}^{-1}(F),$ where the function $\pi
_{(t_{1},...,t_{n})}:\Lambda \rightarrow $ $\Lambda _{t_{1}}\times \cdots
\times \Lambda _{t_{n}}$ is the canonical projection on $\Lambda $ defined
by the relations 
\begin{eqnarray}
\pi _{(t_{1},...,t_{n})}(\lambda ) &:&=(\pi _{t_{1}}(\lambda ),...,\pi
_{t_{n}}(\lambda ))\in \Lambda _{t_{1}}\times \cdots \times \Lambda _{t_{n}},
\label{17} \\
\pi _{t}(\lambda ) &:&=\lambda _{t}\in \Lambda _{t},  \notag
\end{eqnarray}%
we have 
\begin{equation}
\mathcal{A}_{\Lambda }=\{\pi _{(t_{1},...,t_{n})}^{-1}(F)\subseteq \Lambda 
\text{ }\mid \text{ }F\in \mathcal{F}_{\Lambda _{t_{1}}\times \cdots \times
\Lambda _{t_{n}}},\text{\ \ }\{t_{1},...,t_{n}\}\subset T,\text{\ \ }n\in 
\mathbb{N}\}.  \label{18}
\end{equation}%
\smallskip

Introduce a family 
\begin{equation}
\{\mathfrak{M}_{(t_{1},...,t_{n})}:\mathcal{F}_{\Lambda _{t_{1}}\times
\cdots \times \Lambda _{t_{n}}}\rightarrow \mathcal{L}_{\mathcal{H}}\mid 
\mathfrak{M}_{(t_{1},...,t_{n})}(\Lambda _{t_{1}}\times \cdots \times
\Lambda _{t_{n}})=\mathbb{I}_{\mathcal{H}},\text{ \ \ }\{t_{1},...,t_{n}\}%
\subset T,\text{ \ \ }n\in \mathbb{N}\}  \label{18_1}
\end{equation}%
of normalized finitely additive measures $\mathfrak{M}_{(t_{1},...,t_{n})}$,
each specified by mutually non-equal indices $t_{1},...,t_{n}$ $\in T$ and
having values that are bounded linear operators on $\mathcal{H}.$

Let, for each finite index collection $\{t_{1},...,t_{n}\}\subset T,$ these
measures satisfy the consistency condition 
\begin{eqnarray}
\mathfrak{M}_{(t_{1},...,t_{n})}(F_{1}\times \cdots \times F_{n}) &=&%
\mathfrak{M}_{(t_{i_{1}},...,t_{i_{_{n}}})}(F_{i_{1}}\times \cdots \times
F_{i_{n}}),  \label{13_} \\
F_{i} &\in &\mathcal{F}_{\Lambda _{t_{i}}},\text{ \ \ }i=1,...,n,  \notag
\end{eqnarray}%
for all permutations $\binom{1,,...,n}{i_{1},...,i_{n}}$ and the consistency
condition 
\begin{eqnarray}
&&\mathfrak{M}_{(t_{1},...,t_{_{n}})}\left( \{(\lambda _{1},...,\lambda
_{n})\in \Lambda _{t_{1}}\times \cdots \times \Lambda _{t_{n}}\mid (\lambda
_{i_{1}},...,\lambda _{i_{k}})\in F\}\right)  \label{14_} \\
&=&\mathfrak{M}_{(t_{i_{1}},...,t_{i_{k}})}\left( F\right) ,\text{ \ \ \ \ \ 
}F\in \mathcal{F}_{\Lambda _{t_{i_{1}}}\times \cdots \times \Lambda
_{t_{i_{_{k}}}}},  \notag
\end{eqnarray}%
for each $\{t_{i_{1}},...,t_{i_{k}}\}\subseteq
\{t_{1},...,t_{n}\}.\smallskip $

The following statement is proved in appendix B and constitutes a
generalization to the case of consistent operator-valued measures of some
items of the Kolmogorov consistency theorem \cite{7, 17} for probability
measures (\ref{15}).

\begin{lemma}
For a family (\ref{18_1}) of normalized finitely additive $\mathcal{L}_{%
\mathcal{H}}$-valued measures $\mathfrak{M}_{(t_{i_{1}},...,t_{i_{n}})}$
satisfying the consistency conditions (\ref{13_}) (\ref{14_}), there exists
a unique normalized finitely additive $\mathcal{L}_{\mathcal{H}}$-valued
measure 
\begin{equation}
\mathbb{M}:\mathcal{A}_{\Lambda }\rightarrow \mathcal{L}_{\mathcal{H}},\ \ \
\ \mathbb{M(}\Lambda )=\mathbb{I}_{\mathcal{H}},  \label{19_}
\end{equation}%
such that 
\begin{equation}
\mathbb{M}\left( \pi _{(t_{1},...,t_{n})}^{-1}(F)\right) =\mathfrak{M}%
_{(t_{1},...,t_{_{n}})}(F)\text{\ }  \label{20_}
\end{equation}%
for all sets $F\in \mathcal{F}_{\Lambda _{t_{1}}\times \cdots \times \Lambda
_{t_{n}}}$ and an arbitrary finite index collection $\{t_{1},...,t_{n}\}%
\subset T.$\medskip
\end{lemma}

Note that family (\ref{12}) of the product measures $\mathcal{P}%
_{(X_{1},...,X_{n})}$ represents a particular example of a family (\ref{18_1}%
) if, in the latter, we replace 
\begin{eqnarray}
T &\rightarrow &\mathfrak{X}_{\mathcal{H}},\text{ \ \ }\Lambda
_{t}\rightarrow \mathrm{sp}X, \\
\mathcal{F}_{t} &\rightarrow &\mathcal{B}_{\mathrm{sp}X},\text{ \ \ \ }%
\mathcal{F}_{\Lambda _{t_{1}}\times \cdots \times \Lambda
_{t_{n}}}\rightarrow \mathcal{F}_{\mathrm{sp}X_{1}\times \cdots \times 
\mathrm{sp}X_{n}}.  \notag
\end{eqnarray}%
Moreover, in view of lemma 1, all measures $\mathcal{P}_{(X_{1},...,X_{n})}$
satisfy the consistency conditions (\ref{13_}), (\ref{14_}).

Therefore, similarly to our notations in lemma 2, we denote by let $%
\widetilde{\Lambda }:=\dprod\limits_{X\in \mathfrak{X}_{\mathcal{H}}}\mathrm{%
sp}X$ the set of all real-valued functions $\widetilde{\lambda }:\mathfrak{X}%
_{\mathcal{H}}\rightarrow \cup _{X\in \mathfrak{X}_{\mathcal{H}}}\mathrm{sp}%
X $ with values $\widetilde{\lambda }_{X}:=\widetilde{\lambda }(X)\in 
\mathrm{sp}X$.

Let 
\begin{eqnarray}
\pi _{(X_{1},...,X_{n})}(\widetilde{\lambda }) &:&=(\pi _{X_{1}}(\widetilde{%
\lambda }),...,\pi _{X_{n}}(\widetilde{\lambda }))\in \mathrm{sp}X_{1}\times
\cdots \times \mathrm{sp}X_{n}\subseteq \mathbb{R}^{n},  \label{17_} \\
\pi _{X}(\widetilde{\lambda }) &:&=\widetilde{\lambda }_{X}\in \mathrm{sp}X.
\notag
\end{eqnarray}%
be the canonical projection $\widetilde{\Lambda }\rightarrow \mathrm{sp}%
X_{1}\times \cdots \times \mathrm{sp}X_{n}$. The set 
\begin{equation}
\mathcal{A}_{\widetilde{\Lambda }}:=\{\pi
_{(t_{1},...,t_{n})}^{-1}(F)\subseteq \widetilde{\Lambda }\text{ }\mid \text{
}F\in \mathcal{F}_{\mathrm{sp}X_{1}\times \cdots \times \mathrm{sp}X_{n}},%
\text{\ \ }\{X_{1},...,X_{n}\}\subset \mathfrak{X}_{\mathcal{H}},\text{\ \ }%
n\in \mathbb{N}\}.  \label{al}
\end{equation}%
of all cylindrical subsets $\pi _{(t_{1},...,t_{n})}^{-1}(F)$ constitutes an
algebra on $\widetilde{\Lambda }.$

Lemma 2 implies.

\begin{theorem}
For family (\ref{12}) of finitely additive measures $\mathcal{P}%
_{(X_{1},...,X_{n})},$ there exists a unique normalized finitely additive $%
\mathcal{L}_{\mathcal{H}}^{(s)}$-valued measure 
\begin{equation}
\mathbb{P}:\mathcal{A}_{\widetilde{\Lambda }}\rightarrow \mathcal{L}_{%
\mathcal{H}}^{(s)},\ \ \ \ \mathbb{P(}\widetilde{\Lambda })=\mathbb{I}_{%
\mathcal{H}},  \label{19}
\end{equation}%
such that 
\begin{equation}
\mathbb{P}\left( \pi _{(X_{1},...,X_{n})}^{-1}(F)\right) =\mathcal{P}%
_{(X_{1},...,X_{n})}(F)  \label{20}
\end{equation}%
for all sets $F\in \mathcal{F}_{\mathrm{sp}X_{1}\times \cdots \times \mathrm{%
sp}X_{n}}$ and an arbitrary finite collection $\{X_{1},...,X_{n}\}$ $\subset 
\mathfrak{X}_{\mathcal{H}}$. In particular,%
\begin{eqnarray}
\mathbb{P(}\pi _{X_{1}}^{-1}(B_{1})\cap \cdots \cap \pi
_{X_{n}}^{-1}(B_{n})) &=&\frac{1}{n!}\left\{ \mathrm{P}_{X_{1}}(B_{1})\cdot
\ldots \cdot \mathrm{P}_{X_{n}}(B_{n})\right\} _{\mathrm{sym}},  \label{21}
\\
B_{1} &\in &\mathcal{B}_{\mathrm{sp}X_{1}},...,B_{n}\in \mathcal{B}_{\mathrm{%
sp}X_{n}}  \notag
\end{eqnarray}%
for each finite number of mutually non-equal operators $X_{1},...,X_{n}\in 
\mathfrak{X}_{\mathcal{H}}.$\medskip
\end{theorem}

Theorem 1 implies.

\begin{proposition}
Let $\{\mathcal{P}_{(X_{1},...,X_{n})}\}$ be $\mathcal{L}_{\mathcal{H}%
}^{(s)} $-valued measures (\ref{me}). For every density operator $\rho $ on $%
\mathcal{H}$, there exists a unique normalized finitely additive real-valued
measure%
\begin{equation}
\mu _{\rho }:\mathcal{A}_{\widetilde{\Lambda }}\rightarrow \mathbb{R},\text{
\ \ }\mu _{\rho }(\widetilde{\Lambda })=1,  \label{22}
\end{equation}%
such that 
\begin{equation}
\mathrm{tr}[\rho \mathcal{P}_{(X_{1},...,X_{n})}(F)]=\mu _{\rho }\left( \pi
_{(X_{1},...,X_{n})}^{-1}(F)\right)  \label{23}
\end{equation}%
for all sets $F\in \mathcal{F}_{\mathrm{sp}X_{1}\times \cdots \times \mathrm{%
sp}X_{n}}$ and an arbitrary finite collection $\{X_{1},...,X_{n}\}$ $\subset 
\mathfrak{X}_{\mathcal{H}}.$ In particular,%
\begin{eqnarray}
\frac{1}{n!}\mathrm{tr}[\rho \{\mathrm{P}_{X_{1}}(B_{1})\cdot \ldots \cdot 
\mathrm{P}_{X_{n}}(B_{n})\}_{\mathrm{sym}}] &=&\mu _{\rho }\left( \pi
_{X_{1}}^{-1}(B_{1})\cap \cdots \cap \pi _{X_{n}}^{-1}(B_{n})\right) ,
\label{24} \\
B_{1} &\in &\mathcal{B}_{\mathrm{sp}X_{1}},...,B_{n}\in \mathcal{B}_{\mathrm{%
sp}X_{n}},  \notag
\end{eqnarray}%
for each finite number of mutually non-equal operators $X_{1},...,X_{n}\in 
\mathfrak{X}_{\mathcal{H}}.$
\end{proposition}

\begin{proof}
For a density operator $\rho $ on $\mathcal{H},$ relation (\ref{20}) implies%
\begin{equation}
\mathrm{tr}[\rho \mathcal{P}_{(X_{1},...,X_{n})}(F)]=\mathrm{tr}[\rho 
\mathbb{P}\left( \pi _{(X_{1},...,X_{n})}^{-1}(F)\right) ]  \label{25}
\end{equation}%
for all sets $F\in \mathcal{F}_{\mathrm{sp}X_{1}\times \cdots \times \mathrm{%
sp}X_{n}}.$ Introduce on the algebra $\mathcal{A}_{\widetilde{\Lambda }}$
the normalized finitely additive real-valued measure 
\begin{equation}
\mu _{\rho }\left( A\right) :=\mathrm{tr}[\rho \mathbb{P}\left( A\right) ],%
\text{ \ \ }A\in \mathcal{A}_{\widetilde{\Lambda }}.  \label{26}
\end{equation}%
Since $\mathbb{P}$ is a unique $\mathcal{L}_{\mathcal{H}}^{(s)}$-valued
finitely additive measure on $\mathcal{A}_{\widetilde{\Lambda }}$ satisfying
condition (\ref{20}), the measure $\mu _{\rho }$ defined by relation (\ref%
{25}) is also a unique normalized real-valued finitely additive measure on $%
\mathcal{A}_{\widetilde{\Lambda }}$ satisfying condition (\ref{23}), hence, (%
\ref{24}).
\end{proof}

\section{Quasi-classical probability modelling}

Based on theorem 1 and proposition 1, we proceed to prove that the Hilbert
space description of all joint von Neumann measurements on a quantum system
can be reproduced via a set of random variables and a set of normalized
real-valued measures on a single space $(\Omega ,\mathcal{F}_{\Omega }).$

\begin{theorem}
Let $\mathcal{H}$ be a complex separable Hilbert space. There
exist:\smallskip \newline
(i) a set $\Omega $ and an algebra $\mathcal{F}_{\Omega }$ of subsets of $%
\Omega ;$ \smallskip \newline
(ii) a $\mathcal{F}_{\Omega }/\mathcal{B}_{\mathrm{sp}X}$-measurable\textrm{%
\ }real-valued function (random variable) $f_{X}:\Omega \rightarrow \mathrm{%
sp}X$ for each quantum observable $X$ on $\mathcal{H}$;\smallskip \newline
such that $f_{X_{1}}\neq f_{X_{2}}$ for $X_{1}\neq X_{2}$ and, to each
quantum state $\rho $ on $\mathcal{H},$ there corresponds a unique
normalized finitely additive real-valued measure $\mu _{\rho }$ on $(\Omega ,%
\mathcal{F}_{\Omega })$ satisfying the relation 
\begin{eqnarray}
\frac{1}{n!}\mathrm{tr}[\rho \{\mathrm{P}_{X_{1}}(B_{1})\cdot \ldots \cdot 
\mathrm{P}_{X_{n}}(B_{n})\}_{\mathrm{sym}}] &=&\mu _{\rho }\left( \text{ }%
f_{X_{1}}^{-1}(B_{1})\cap \cdots \cap f_{X_{n}}^{-1}(B_{n})\right) ,
\label{27} \\
B_{1} &\in &\mathcal{B}_{\mathrm{sp}X_{1}},...,B_{n}\in \mathcal{B}_{\mathrm{%
sp}X_{n}},  \notag
\end{eqnarray}%
\smallskip for each finite collection $\{X_{1},...,X_{n}\}$ of quantum
observables on $\mathcal{H}$. In particular,%
\begin{eqnarray}
\mathrm{tr}[\rho \{\mathrm{P}_{X_{1}}(B_{1})\cdot \ldots \cdot \mathrm{P}%
_{X_{n}}(B_{n})\}] &=&\mu _{\rho }\left( \text{ }f_{X_{1}}^{-1}(B_{1})\cap
\cdots \cap f_{X_{n}}^{-1}(B_{n})\right) ,  \label{28} \\
B_{1} &\in &\mathcal{B}_{\mathrm{sp}X_{1}},...,B_{n}\in \mathcal{B}_{\mathrm{%
sp}X_{n}},  \notag
\end{eqnarray}%
for every state $\rho $ and an arbitrary finite collection $%
\{X_{1},...,X_{n}\}$ of quantum observables on $\mathcal{H}$ mutually
commuting in the sense of relation (\ref{com}).
\end{theorem}

\begin{proof}
In order to prove the existence point of theorem 2, let us take the space $(%
\widetilde{\Lambda },\mathcal{A}_{\widetilde{\Lambda }})$ considered in
theorem 1. Namely, $\widetilde{\Lambda }$ is the set of all real-valued
functions $\widetilde{\lambda }:\mathfrak{X}_{\mathcal{H}}\rightarrow \cup
_{X\in \mathfrak{X}_{\mathcal{H}}}\mathrm{sp}X$ with values $\widetilde{%
\lambda }_{X}\in \mathrm{sp}X$ and $\mathcal{A}_{\widetilde{\Lambda }}$ is
the algebra (\ref{al}) of all cylindrical subsets of $\widetilde{\Lambda }$
having the form $\pi _{(X_{1},...,X_{n})}^{-1}(F),$ where $F\in \mathcal{F}_{%
\mathrm{sp}X_{1}\times \cdots \times \mathrm{sp}X_{n}}$ and $%
\{X_{1},...,X_{n}\}$ $\subset \mathfrak{X}_{\mathcal{H}}.$

For each observable $X\in \mathfrak{X}_{\mathcal{H}}$, we take on $(%
\widetilde{\Lambda },\mathcal{A}_{\widetilde{\Lambda }})$ the random
variable $\pi _{X}(\widetilde{\lambda })=\widetilde{\lambda }_{X}\in \mathrm{%
sp}X$ $\subseteq R$ and note that $\pi _{X_{1}}\neq \pi _{X_{2}}$ for $%
X_{1}\neq X_{2}$.

Then, by proposition 2, to each quantum state $\rho $ on $\mathcal{H},$
there corresponds a unique normalized real-valued measures $\mu _{\rho }$ on 
$\mathcal{A}_{\widetilde{\Lambda }}$ satisfying (\ref{24}) and, hence,
relations (\ref{27}) and (\ref{28}). \medskip
\end{proof}

From relations (\ref{27}) and (\ref{2}) it follows.

\begin{corollary}
In theorem 2, let $\{X_{1},...,X_{n}\}$ be a finite collection of bounded
quantum observables $X_{1},...,X_{n}$ on $\mathcal{H}$. Then 
\begin{equation}
\frac{1}{n!}\mathrm{tr}[\rho \{X_{1}\cdot \ldots \cdot X_{n}\}_{\mathrm{sym}%
}]=\dint\limits_{\Omega }f_{X_{1}}(\omega )\cdot \ldots \cdot
f_{X_{n}}(\omega )\text{ }\mu _{\rho }\left( \mathrm{d}\omega \right)
\label{29}
\end{equation}%
for all quantum states $\rho .$
\end{corollary}

Theorem 2 implies.

\begin{corollary}
For the probabilistic description of all joint von Neumann measurements upon
a quantum state $\rho $ on a complex separable Hilbert space $\mathcal{H},$
there exist:\smallskip \newline
(i) a measure space $(\Omega ,\mathcal{F}_{\Omega },\mu _{\rho }),$ where $%
\mathcal{F}_{\Omega }$ is an algebra of subsets of a set $\Omega $ and $\mu
_{\rho }$ is a normalized finitely additive real-valued measure on $\mathcal{%
F}_{\Omega }$;\newline
(ii) a random variable $f_{X}:\Omega \rightarrow \mathrm{sp}X$ for each
quantum observable $X$ on $\mathcal{H};\smallskip $\newline
such that $f_{X_{1}}\neq f_{X_{2}}$ for $X_{1}\neq X_{2},$ a space $(\Omega ,%
\mathcal{F}_{\Omega })$ and random variables $\{f_{X}\}$ do not depend on a
state $\rho $ and the representation 
\begin{eqnarray}
\mathrm{tr}[\rho \{\mathrm{P}_{X_{1}}(B_{1})\cdot \ldots \cdot \mathrm{P}%
_{X_{n}}(B_{n})\}] &=&\mu _{\rho }\left( \text{ }f_{X_{1}}^{-1}(B_{1})\cap
\cdots \cap f_{X_{n}}^{-1}(B_{n})\right) ,  \label{28'} \\
B_{1} &\in &\mathcal{B}_{\mathrm{sp}X_{1}},...,B_{n}\in \mathcal{B}_{\mathrm{%
sp}X_{n}},  \notag
\end{eqnarray}%
\smallskip holds for an arbitrary finite collection $\{X_{1},...,X_{n}\}$ of
quantum observables on $\mathcal{H}$ mutually commuting in the sense of
relation (\ref{com}).
\end{corollary}

From corollary 2 it follows that the probability distributions of \emph{all} 
\emph{joint} von Neumann measurements on a quantum state $\rho $ can be
reproduced in terms of a single measure space $(\Omega ,\mathcal{F}_{\Omega
},\mu _{\rho })$ with a normalized real-valued measure $\mu _{\rho }$ and a
set of random variables, each modelling the corresponding quantum observable
in all joint von Neumann measurements and depending only on this quantum
observable.

Representation (\ref{28'}) can be otherwise expressed in the form%
\begin{equation}
\mathrm{tr}[\rho \{\mathrm{P}_{X_{1}}(B_{1})\cdot \ldots \cdot \mathrm{P}%
_{X_{n}}(B_{n})\}]=\dint\limits_{\Omega }\chi
_{f_{X_{1}}^{-1}(B_{1})}(\omega )\cdot \ldots \cdot \chi
_{f_{X_{n}}^{-1}(B_{n})}(\omega )\mu _{\rho }\left( \text{d}\omega \right) ,
\label{30'}
\end{equation}%
which is specific for joint probability distributions in a local quasi
hidden variable (LqHV) model of the deterministic type, see Refs. \cite{11,
12}.

Thus, all joint von Neumann measurements on a finite dimensional quantum
state admit \emph{a deterministic quasi hidden variable (qHV) model} \cite%
{11, 12} and, in this model, a random variable modelling a quantum
observable depends only on this quantum observable.

The following statement is proved in appendix C.

\begin{proposition}
In theorem 2:\newline
(i) If $f_{X}:\Omega \rightarrow \mathrm{sp}X$ is a random variable
modelling a quantum observable $X$ via representation (\ref{27}), then, for
each Borel function $\varphi :\mathbb{R}\rightarrow \mathbb{R},$ the random
variable $\varphi \circ f_{X}:\Omega \rightarrow \mathrm{sp}\varphi (X)$
models the quantum observable $\varphi (X);$\newline
(ii) If $\mu _{\rho _{k}},$ $k=1,...,K<\infty ,$ are normalized real-valued
measures, each uniquely corresponding to a quantum state $\rho _{k}$ via
representation (\ref{27}), then the measure $\sum \alpha _{k}\mu _{\rho
_{k}},$ with $\alpha _{k}>0,$ $\sum \alpha _{k}=1,$ uniquely corresponds to
the state $\sum \alpha _{k}\rho _{k}.\medskip $
\end{proposition}

From theorem 2 and proposition 2 it follows that the observable $\varphi (X)$
is modeled via representation (\ref{27}) by either of two random variables $%
\varphi \circ f_{X}$ or $f_{\varphi (X)}$ on $(\Omega ,\mathcal{F}_{\Omega
}) $ and, for arbitrary $X$ and $\varphi ,$ the latter random variables do
not need to coincide.

Consider, for example, the random variables $\pi _{X},$ $X\in \mathfrak{X}_{%
\mathcal{H}},$ defined on the space $(\widetilde{\Lambda },\mathcal{A}_{%
\widetilde{\Lambda }})$ by relation (\ref{17_}) and used by us above for the
proof of the existence point of theorem 2. The random variables $\varphi
\circ \pi _{X}$ and $\pi _{\varphi (X)}$ do not need to \ coincide for all
observable $X$ and all Borel functions $\varphi :\mathbb{R}\rightarrow 
\mathbb{R}.$ Hence, for arbitrary $X$ and $\varphi ,$ the observable $%
\varphi (X)$ is equivalently modeled on the space $(\widetilde{\Lambda },%
\mathcal{A}_{\widetilde{\Lambda }})$ by either of two different random
variables -- $\varphi \circ \pi _{X}$ and $\pi _{\varphi (X)}.$

Note also that, according to the Kochen and Specker theorem \cite{5}, for a
Hilbert space $\mathcal{H}$ of a dimension $d\geq 3,$ there does not exist a
space $(\Omega ,\mathcal{F}_{\Omega }),$ where, under the condition $\varphi
\circ f\overset{\Phi }{\mapsto }\varphi \circ X,$ $\forall \varphi ,$ a
mapping $f\overset{\Phi }{\mapsto }X$ from a set of random variables on $%
(\Omega ,\mathcal{F}_{\Omega })$ onto the set of all quantum observables on $%
\mathcal{H}$ could be one-to-one.\textbf{\ }

\section{Conclusions}

In the present paper, we have introduced (lemma 2, theorem 1) a
generalization of some items of the Kolmogorov extension theorem \cite{7, 17}
to the case of consistent operator-valued measures and, based on this, we
have proved (theorem 2) that the Hilbert space description of all joint von
Neumann measurements on a quantum system can be reproduced in terms of a
single space $(\Omega ,\mathcal{F}_{\Omega })$ via a set of normalized
real-valued measures, each uniquely corresponding to some quantum state, and
a set of random variables, each being determined only by a modelled quantum
observable and such that if $f_{X}$ is a random variable modelling a quantum
observable $X$ via representation (\ref{27}), then, for each Borel function $%
\varphi :\mathbb{R}\rightarrow \mathbb{R},$ the random variable $\varphi
\circ X$ models (proposition 2) the observable $\varphi \circ X$.

This result, in particular, means that \emph{all joint von Neumann
measurements} on a quantum state\emph{\ }$\rho $ admit (corollary 2) \emph{a
deterministic quasi hidden variable (qHV) model} \cite{11, 12} and, in this
model, a random variable modelling a quantum observable depends only on this
quantum observable.

From the probabilistic point of view, a deterministic qHV model constitutes 
\emph{the} \emph{quasi-classical probability model }-- a new general
probability model formulated in Ref. \cite{12}.

In the quasi-classical probability model specified by a measure space $%
(\Omega ,\mathcal{F}_{\Omega },\nu ),$ a normalized real-valued measure $\nu 
$ does not need to be positive but:\smallskip

\noindent (i) an observable with a value space $(\Lambda ,\mathcal{F}%
_{\Lambda })$ is represented only by such a random variable $f:\Omega
\rightarrow \Lambda ,$ for which the normalized measure $\nu (f^{-1}(\cdot
)) $ on $\mathcal{F}_{\Lambda }$ is a probability one;\smallskip

\noindent (ii) a joint measurement of two observables represented by random
variables $f_{1},$ $f_{2}$ with value spaces $(\Lambda _{n},\mathcal{F}%
_{\Lambda _{n}}),$ $n=1,2,$ is possible iff $\nu (f_{1}^{-1}(F_{1})\cap
f_{2}^{-1}(F_{2}))\geq 0,$ for all $F_{1}\in \mathcal{F}_{\Lambda _{1}},$ $%
F_{2}\in \mathcal{F}_{\Lambda _{2}}.$\smallskip

We stress that though, in the quasi-classical probability model, a measure
space $(\Omega ,\mathcal{F}_{\Omega },\nu )$ does not need to be a
probability one, each modelled measurement, single or joint, satisfies the
Kolmogorov axioms \cite{7} in the sense that it is described by a
probability space, where a probability measure is specified in the above
item (i) or (ii).

From the above results it also follows that if joint von Neumann
measurements on an $N$-partite quantum state $\rho $ are performed by
space-like separated parties, then, for these joint measurements, there
exists a quasi-classical probability model $(\Omega ,\mathcal{F}_{\Omega
},\mu _{\rho })$, which is \emph{local} in the sense that each party
marginal measurement is described by a random variable depending only on the
corresponding observable at the corresponding site, but does not need to be
"classical" -- in the sense of positivity a measure $\mu _{\rho }.$

From the conceptual point of view, the latter mathematical result not only
supports our arguments \cite{10, 18} on a difference between Bell's locality
and the EPR locality but also directly points to a misleading character of
Bell's conjecture on quantum "non-locality" (\emph{action on a distance})
that was introduced by Bell \cite{4, 3} only in view of non-existence of a
local classical probability model for spin measurements on the two-qubit
singlet and due to his further choice that this non-existence is caused
precisely by violation of "locality" but not by violation of "classicality".

\section{Appendix A}

Relation (\ref{13}) follows explicitly from the symmetrized form of the
right-hand side of condition (\ref{9_}).

For clearness, let us first prove relation (\ref{14}) for a collection $%
\{X_{1},...,X_{n}\}$ of bounded quantum observables with discrete spectrums.
In this case, the measure $\mathcal{P}_{(X_{1},...,X_{_{n}})}$ is given by
representation (\ref{9}) and taking into the account the relation $\mathrm{P}%
_{X_{i}}(\mathrm{sp}X_{i})=\mathbb{I}_{\mathcal{H}}$, we have:%
\begin{eqnarray}
&&\mathcal{P}_{(X_{1},...,X_{_{n}})}\left( \{(x_{1},...,x_{n})\in \mathrm{sp}%
X_{1}\times \cdots \times \mathrm{sp}X_{n}\mid (x_{i_{1}},...,x_{i_{k}})\in
F\}\right)  \TCItag{A1}  \label{A1} \\
&=&\frac{1}{n!}\dsum\limits_{(\lambda _{X_{i_{1}}},...,\lambda
_{X_{i_{k}}})\in F}\text{ }\left\{ \text{ }\mathrm{P}_{X_{1}}(\{\lambda
_{X_{1}}\})\cdot \ldots \cdot \mathrm{P}_{X_{n}}(\{\lambda
_{X_{n}}\})\right\} _{\mathrm{sym}}  \notag \\
&=&\frac{1}{k!}\dsum\limits_{(\lambda _{X_{i_{1}}},...,\lambda
_{X_{i_{k}}})\in F}\left\{ \mathrm{P}_{X_{i_{1}}}(\{\lambda
_{X_{i_{1}}}\})\cdot \ldots \cdot \mathrm{P}_{X_{i_{k}}}(\{\lambda
_{X_{i_{k}}}\})\right\} _{\mathrm{sym}}  \notag \\
&=&\mathcal{P}_{(X_{i_{1}},...,X_{i_{k}})}\left( F\right) .  \notag
\end{eqnarray}

In order to prove (\ref{14}) for an arbitrary collection $%
\{X_{1},...,X_{n}\} $ $\subset \mathfrak{X}_{\mathcal{H}}$ of quantum
observables, let us denote by $\mathcal{E}$ the set of all rectangles $%
E:=B_{i_{1}}\times \cdots \times B_{i_{k}}$ with sides $B_{i}\in \mathcal{B}%
_{\mathrm{sp}X_{i}}.$ Since the algebra $\mathcal{F}_{\mathrm{sp}%
X_{i_{1}}\times \cdots \times \mathrm{sp}X_{i_{k}}}$ consists of all finite
unions of mutually disjoint rectangles from $\mathcal{E}$, every $F\in 
\mathcal{F}_{\mathrm{sp}X_{i_{1}}\times \cdots \times \mathrm{sp}X_{i_{k}}}$
admits a finite decomposition 
\begin{equation}
F=\underset{m=1,...,M}{\cup }E_{m},\text{ \ \ }E_{m_{1}}\cap
E_{m_{2}}=\varnothing ,\ \ E_{m}\in \mathcal{E},\text{ \ }M<\infty . 
\tag{A2}  \label{A2}
\end{equation}%
Taking into the account that $\mathcal{P}_{(X_{1},...,X_{_{n}})},$ $\mathcal{%
P}_{(X_{i_{1}},...,X_{i_{k}})}$ are measures and relations (\ref{A2}), (\ref%
{9_}) and $\mathrm{P}_{X_{i}}(\mathrm{sp}X_{i})=\mathbb{I}_{\mathcal{H}},$
we derive:%
\begin{eqnarray}
&&\mathcal{P}_{(X_{1},...,X_{_{n}})}\left( \{(x_{1},...,x_{n})\in \mathrm{sp}%
X_{1}\times \cdots \times \mathrm{sp}X_{n}\mid (x_{i_{1}},...,x_{i_{k}})\in
F\}\right)  \TCItag{A3}  \label{A3} \\
&=&\sum_{m}\mathcal{P}_{(X_{1},...,X_{_{n}})}\left( \{(x_{1},...,x_{n})\in 
\mathrm{sp}X_{1}\times \cdots \times \mathrm{sp}X_{n}\mid
(x_{i_{1}},...,x_{i_{k}})\in E_{m}\}\right)  \notag \\
&=&\sum_{m}\frac{1}{k!}\left\{ \mathrm{P}_{X_{i_{1}}}(B_{i_{1}}^{(m)})\cdot
\ldots \cdot \mathrm{P}_{X_{i_{k}}}(B_{i_{k}}^{(m)})\right\} _{\mathrm{sym}%
}=\sum_{m}\mathcal{P}_{(X_{i_{1}},...,X_{i_{k}})}\left( E_{m}\right)  \notag
\\
&=&\mathcal{P}_{(X_{i_{1}},...,X_{i_{k}})}\left( F\right) .  \notag
\end{eqnarray}

This prove lemma 1.

\section{Appendix B}

Our proof of lemma 2 is quite similar to the proof \cite{7, 17} of the
corresponding items in the Kolmogorov extension theorem for consistent
probability measures (\ref{15}).

Let $\mathcal{A}_{\Lambda }$ be the algebra on $\Lambda $ defined by
relations (\ref{16}) - (\ref{18}). For a set $A\in \mathcal{A}_{\Lambda }$
admitting the representation $A=\pi _{(t_{1},...,t_{n})}^{-1}(F),$ where $%
\{t_{1},...,t_{n}\}\subset T$ and $F\in \mathcal{F}_{\Lambda _{t_{1}}\times
\cdots \times \Lambda _{t_{n}}}$, we let 
\begin{subequations}
\begin{equation}
\mathbb{M}(A):=\mathfrak{M}_{(t_{1},...,t_{_{n}})}(F).  \tag{B1}  \label{B1}
\end{equation}

In order to show that relation (\ref{B1}) defines correctly a set function
on $\mathcal{A}_{\Lambda },$ we must prove that this relation implies a
unique value to a set $A\in \mathcal{A}_{\Lambda }$ even if this set admit
two different representations, say: 
\end{subequations}
\begin{eqnarray}
A &=&\pi _{(t_{i_{1}},...,t_{i_{k}})}^{-1}(F)\equiv \left\{ \lambda \in
\Lambda \mid (\lambda _{t_{i_{1}}},...,\lambda _{t_{i_{k}}})\in F\right\} , 
\TCItag{B2}  \label{B2} \\
A &=&\pi _{(t_{j_{1}},...,t_{j_{_{m}}})}^{-1}(F^{\prime })\equiv \left\{
\lambda \in \Lambda \mid (\lambda _{t_{j_{1}}},...,\lambda _{t_{j_{m}}})\in
F^{\prime }\right\}  \notag
\end{eqnarray}%
for some sets $F\in \mathcal{F}_{\Lambda _{t_{i_{1}}}\times \cdots \times
\Lambda _{t_{i_{k}}}}$ and $F^{\prime }\in \mathcal{F}_{\Lambda
_{t_{j_{1}}}\times \cdots \times \Lambda _{t_{j_{m}}}}$ and some index
collections $\{t_{i_{1}},...,t_{i_{k}}\},$ $\{t_{j_{1}},...,t_{j_{_{m}}}\}%
\subset T$ .

Denote 
\begin{equation}
\{t_{i_{1}},...,t_{i_{k}}\}\cup
\{t_{j_{1}},...,t_{j_{_{m}}}\}:=\{t_{1},...,t_{_{n}}\}.  \tag{B3}  \label{B3}
\end{equation}

From (\ref{B2}) it follows that sets $F$ and $F^{\prime }$ are such that,
for a point in $\Lambda _{t_{1}}\times \cdots \times \Lambda _{t_{n}},$ the
condition $(\lambda _{i_{1}},...,\lambda _{i_{k}})\in F$ implies the
condition $(\lambda _{j_{1}},...,\lambda _{j_{m}})\in F^{\prime }$ and vice
versa, that is:%
\begin{eqnarray}
&&\left\{ \text{ }(\lambda _{1},...,\lambda _{n})\in \Lambda _{t_{1}}\times
\cdots \times \Lambda _{t_{n}}\mid (\lambda _{i_{1}},...,\lambda
_{i_{k}})\in F\right\}  \TCItag{B4}  \label{B4} \\
&=&\left\{ \text{ }(\lambda _{1},...,\lambda _{n})\in \Lambda _{t_{1}}\times
\cdots \times \Lambda _{t_{n}}\mid (\lambda _{j_{1}},...,\lambda
_{j_{m}})\in F^{\prime }\right\} .  \notag
\end{eqnarray}

In view of relations (\ref{B1}) - (\ref{B4}) and the consistency conditions (%
\ref{13_}), (\ref{14_}), we have:%
\begin{eqnarray}
\mathbb{M}(\pi _{(t_{i_{1}},...,t_{i_{k}})}^{-1}(F)) &=&\mathfrak{M}%
_{(t_{i_{1}},...,t_{i_{k}})}(F)  \TCItag{B5}  \label{B5} \\
&=&\mathfrak{M}_{(t_{1},...,t_{_{n}})}\left( \left\{ (\lambda
_{1},...,\lambda _{n})\in \Lambda _{t_{1}}\times \cdots \times \Lambda
_{t_{n}}\mid (\lambda _{i_{1}},...,\lambda _{i_{k}})\in F\right\} \right) 
\notag \\
&=&\mathfrak{M}_{(t_{1},...,t_{_{n}})}\left( \left\{ (\lambda
_{1},...,\lambda _{n})\in \Lambda _{t_{1}}\times \cdots \times \Lambda
_{t_{n}}\mid (\lambda _{j_{1}},...,\lambda _{j_{m}})\in F^{\prime }\right\}
\right)  \notag \\
&=&\mathfrak{M}_{(t_{j_{1}},...,t_{j_{m}})}(F^{\prime })  \notag \\
&=&\mathbb{M(}\pi _{(t_{j_{1}},...,t_{j_{_{m}}})}^{-1}(F^{\prime })).  \notag
\end{eqnarray}

Thus, relation (\ref{B1}) defines a unique set function $\mathbb{M}:\mathcal{%
A}_{\Lambda }\rightarrow \mathcal{L}_{\mathcal{H}}$ satisfying condition (%
\ref{19_}).

Since $\Lambda =\pi _{(t_{1},...,t_{n})}^{-1}(\Lambda _{t_{1}}\times \cdots
\times \Lambda _{t_{n}})$ and $\mathfrak{M}_{(t_{1},...,t_{n})}(\Lambda
_{t_{1}}\times \cdots \times \Lambda _{t_{n}})=\mathbb{I}_{\mathcal{H}}$,
from (\ref{B1}) it follows that the set function $\mathbb{M}$ is normalized,
that is, $\mathbb{M}(\Lambda )=\mathbb{I}_{\mathcal{H}}.$

In order to prove that the normalized set function $\mathbb{M}:\mathcal{A}%
_{\Lambda }\rightarrow \mathcal{L}_{\mathcal{H}}$ is additive, let us
consider in the algebra $\mathcal{A}_{\Lambda }$ two disjoint sets 
\begin{equation}
A_{1}=\pi _{(t_{i_{1}},...,t_{i_{k}})}^{-1}(F_{1}),\text{ \ \ \ }A_{2}=\pi
_{(t_{j_{1}},...,t_{j_{m}})}^{-1}(F_{2}),  \tag{B6}  \label{B6}
\end{equation}%
specified by some index collections $\{t_{i_{1}},...,t_{i_{k}}\},$ $%
\{t_{j_{1}},...,t_{j_{m}}\}\subseteq $ $\{t_{1},...,t_{n}\}\subset T$ and
sets $F_{1}\in \mathcal{F}_{\Lambda _{t_{i_{1}}}\times \cdots \times \Lambda
_{t_{i_{k}}}}$ and $F_{2}\in \mathcal{F}_{\Lambda _{t_{j_{1}}}\times \cdots
\times \Lambda _{t_{j_{m}}}}.$

Since $A_{1}\cap A_{2}$ $=\varnothing ,$ the sets $F_{1},$ $F_{2}$ in (\ref%
{B6}) are such that, for a point in $\Lambda _{t_{1}}\times \cdots \times
\Lambda _{t_{n}},$ conditions $(\lambda _{i_{1}},...,\lambda _{i_{k}})\in
F_{1}$ and $(\lambda _{j_{1}},...,\lambda _{j_{m}})\in F_{2}$ are mutually
exclusive, that is: 
\begin{eqnarray}
&&\left\{ (\lambda _{1},...,\lambda _{n})\in \Lambda _{t_{1}}\times \cdots
\times \Lambda _{t_{n}}\mid (\lambda _{i_{1}},...,\lambda _{i_{k}})\in
F_{1}\right\}  \TCItag{B7}  \label{B7} \\
&&\cap \left\{ (\lambda _{1},...,\lambda _{n})\in \Lambda _{t_{1}}\times
\cdots \times \Lambda _{t_{n}}\mid (\lambda _{j_{1}},...,\lambda
_{j_{m}})\in F_{2}\right\}  \notag \\
&=&\varnothing .  \notag
\end{eqnarray}

Taking into the account relations (\ref{B1}), (\ref{B6}), (\ref{B7}), the
consistency conditions (\ref{13_}), (\ref{14_}) and also that each $%
\mathfrak{M}_{(t_{1},...,t_{n})}$ is a finitely additive measure, we derive%
\begin{eqnarray}
&&\mathbb{M(}A_{1}\cup A_{2})  \TCItag{B8}  \label{B8} \\
&=&\mathfrak{M}_{(t_{1},...,t_{n})}(\{(\lambda _{1},...,\lambda _{n})\in
\Lambda _{t_{1}}\times \cdots \times \Lambda _{t_{n}}\mid (\lambda
_{i_{1}},...,\lambda _{i_{k}})\in F_{1}\text{ \ or \ }(\lambda
_{j_{1}},...,\lambda _{j_{m}})\in F_{2}\})  \notag \\
&=&\mathfrak{M}_{(t_{1},...,t_{n})}\left( \left\{ (\lambda _{1},...,\lambda
_{n})\in \Lambda _{t_{1}}\times \cdots \times \Lambda _{t_{n}}\mid (\lambda
_{i_{1}},...,\lambda _{i_{k}})\in F_{1}\right\} \right)  \notag \\
&&+\mathfrak{M}_{(t_{1},...,t_{n})}\left( \left\{ (\lambda _{1},...,\lambda
_{n})\in \Lambda _{t_{1}}\times \cdots \times \Lambda _{t_{n}}\mid (\lambda
_{j_{1}},...,\lambda _{j_{m}})\in F_{2}\right\} \right)  \notag \\
&=&\mathfrak{M}_{(t_{i_{1}},...,t_{i_{k}})}\left( F_{1}\right) +\mathfrak{M}%
_{(t_{j_{1}},...,t_{j_{m}})}\left( F_{2}\right)  \notag \\
&=&\mathbb{M(}A_{1})+\mathbb{M(}A_{2}).  \notag
\end{eqnarray}%
Hence, the normalized set function $\mathbb{M}$ on $\mathcal{A}_{\Lambda }$
defined by relation (\ref{B1}) is additive and is, therefore, a finitely
additive measure on $\mathcal{A}_{\Lambda }$, see remark 1.

Thus, we have proved that the set function $\mathbb{M}:\mathcal{A}_{\Lambda
}\rightarrow \mathcal{L}_{\mathcal{H}}$ defined by relation (\ref{B1})
constitutes a unique normalized finitely additive $\mathcal{L}_{\mathcal{H}}$%
-valued measure on the algebra $\mathcal{A}_{\Lambda }$ satisfying relation (%
\ref{19_}).

\section{Appendix C}

For a complex Hilbert space $\mathcal{H},$ let $(\Omega ,\mathcal{F}_{\Omega
})$ be specified in theorem 2. Then the representation%
\begin{eqnarray}
\frac{1}{n!}\mathrm{tr}[\rho \{\mathrm{P}_{X_{1}}(B_{1})\cdot \ldots \cdot 
\mathrm{P}_{X_{n}}(B_{n})\}_{\mathrm{sym}}] &=&\mu _{\rho }\left( \text{ }%
f_{X_{1}}^{-1}(B_{1})\cap \cdots \cap f_{X_{n}}^{-1}(B_{n})\right) , 
\TCItag{C1}  \label{C1} \\
B_{1} &\in &\mathcal{B}_{\mathrm{sp}X_{1}},...,B_{n}\in \mathcal{B}_{\mathrm{%
sp}X_{n}},  \notag
\end{eqnarray}%
holds for all states $\rho $ and an arbitrary finite number of mutually
non-equal quantum observables $X_{1},...,X_{n}$ on $\mathcal{H}.$

From (\ref{C1}) and the relations 
\begin{eqnarray}
\varphi \circ X &\equiv &\varphi (X):=\int \varphi (\lambda )\mathrm{P}_{X}(%
\mathrm{d}\lambda ),  \TCItag{C2}  \label{C2} \\
\mathrm{P}_{\varphi (X)}(B) &=&\mathrm{P}_{X}(\varphi ^{-1}(B)),\text{ \ \ }%
B\in \mathcal{B}_{\mathrm{sp}\varphi (X)},  \notag \\
\mathrm{sp}\varphi (X) &=&\varphi (\mathrm{sp}X),  \notag
\end{eqnarray}%
it follows that, for a Borel function $\varphi :\mathbb{R}\rightarrow 
\mathbb{R}$ and an observable $X_{1},$ the relation 
\begin{eqnarray}
&&\frac{1}{n!}\mathrm{tr}[\rho \{\mathrm{P}_{\varphi (X_{1})}(B_{1})\cdot 
\mathrm{P}_{X_{2}}(B_{2})\cdot \ldots \cdot \mathrm{P}_{X_{n}}(B_{n})\}_{%
\mathrm{sym}}]  \TCItag{C3}  \label{C3} \\
&=&\frac{1}{n!}\mathrm{tr}[\rho \{P_{X_{1}}(\varphi ^{-1}(B_{1}))\cdot
\ldots \cdot \mathrm{P}_{X_{n}}(B_{n})\}_{\mathrm{sym}}]  \notag \\
&=&\mu _{\rho }\left( \text{ }f_{X}^{-1}(\varphi ^{-1}(B_{1}))\cap \cdots
\cap f_{X_{n}}^{-1}(B_{n})\right)  \notag \\
&=&\mu _{\rho }\left( (\varphi \circ \text{ }f_{X}^{-1})(B_{1})\cap \cdots
\cap f_{Y_{n}}^{-1}(B_{n})\right)  \notag
\end{eqnarray}%
is valid for each Borel function $\varphi :\mathbb{R}\rightarrow \mathbb{R},$
all sets $B_{1}\in \mathcal{B}_{\mathrm{sp}\varphi (X)},...,B_{n}\in 
\mathcal{B}_{\mathrm{sp}X_{n}},$ all states $\rho $ and an arbitrary finite
number of mutually non-equal quantum observables $X_{2},,...,X_{n}$ on $%
\mathcal{H}$. By theorem 2, this proves property (i).

Further, let $\rho _{k}\mapsto \mu _{\rho _{k}},$ $k=1,...,K<\infty .$ From (%
\ref{27}) it follows 
\begin{eqnarray}
&&\frac{1}{n!}\mathrm{tr}[(\sum_{k}\alpha _{k}\rho _{k})\{\mathrm{P}%
_{X_{1}}(B_{1})\cdot \ldots \cdot \mathrm{P}_{X_{n}}(B_{n})\}_{\mathrm{sym}}]
\TCItag{C4}  \label{C4} \\
&=&(\sum_{k}\alpha _{k}\mu _{\rho _{k}})\left( \text{ }f_{X_{1}}^{-1}(B_{1})%
\cap \cdots \cap f_{X_{n}}^{-1}(B_{n})\right)  \notag
\end{eqnarray}%
for all $B_{1}\in \mathcal{B}_{\mathrm{sp}X_{1}},...,B_{n}\in \mathcal{B}_{%
\mathrm{sp}X_{n}}$ and each finite collection $\{X_{1},...,X_{n}\}$ of
quantum observables on $\mathcal{H}.$ By theorem 2, this proves property
(ii).

\end{document}